\documentclass[10pt]{article}
\usepackage[T1]{fontenc}
\usepackage[ansinew]{inputenc}
\usepackage[english]{babel}
\usepackage{graphicx}
\usepackage{booktabs}
\usepackage{amsthm}
\usepackage{amsmath}
\usepackage{bbm}
\usepackage{mathtools}
\usepackage{color}
\usepackage{float}

\newcommand{\N}{\mathbbm{N}}

\newcommand{\R}{\mathbbm{R}}
\newcommand{\devec}{\mathbf{R}}
\renewcommand{\P}{\mathbbm{P}}
\newcommand{\T}{\mathbf{T}}
\newcommand{\E}{\mathbf{E}}

\newcommand{\trace}{\mathrm{trace}}
\newcommand{\rmatrix}[1]{\begin{bmatrix}#1\end{bmatrix}}

\newtheorem{example}{Example}
\newtheorem{corollary}{Corollary}
\newtheorem{problem}{Problem}
\newtheorem{remark}{Remark}
\newtheorem{proposition}{Proposition}
\newtheorem{theorem}{Theorem}

\begin{document}
\title{\LARGE \bf Kernel-based system identification\\ from noisy and
  incomplete input-output data}
\author{Riccardo S. Risuleo, Giulio Bottegal and H\r akan Hjalmarsson
  \thanks{R. S. Risuleo, G. Bottegal and H. Hjalmarsson are with the ACCESS
    Linnaeus Center, School of Electrical Engineering, KTH Royal Institute of
    Technology, Sweden ({\tt \small risuleo;bottegal;hjalmars@kth.se}).  This
    work was supported by the European Research Council under the advanced
    grant LEARN, contract 267381 and by the Swedish Research Council under
    contract 621-2009-4017.  }} 
\maketitle

\begin{abstract}
  In this contribution, we propose a kernel-based method for the identification
  of linear systems from noisy and incomplete input-output datasets. We model
  the impulse response of the system as a Gaussian process whose covariance
  matrix is given by the recently introduced stable spline kernel. We adopt an
  empirical Bayes approach to estimate the posterior distribution of the impulse
  response given the data.  The noiseless and missing data samples, together
  with the kernel hyperparameters, are estimated maximizing the joint marginal
  likelihood of the input and output measurements. To compute the
  marginal-likelihood maximizer, we build a solution scheme based on the
  Expectation-Maximization method. Simulations on a benchmark dataset show the
  effectiveness of the method.
\end{abstract}

\section{Introduction}
Common formulations of system identification problems postulate the perfect
knowledge of the input signal feeding the unknown system~\cite{ljung1999system}.
In many applications however, the input signal is available only in a noisy
version, giving raise to a setup usually referred to as an
\emph{errors-in-variables} (EIV) model~\cite{soederstroem2007errors}.  Static
EIV models have been subject of extensive studies in the statistical literature
since the beginning of the last century~\cite{frisch1934statistical}; later the
system identification community has become interested in dynamical EIV
models~\cite{anderson1985identification,fernando1985identification,soederstroem2007errors}.

Identification of EIV systems is a challenging task; even in the linear case,
standard least-squares yields biased estimates, due to the presence of noise in
the regressors. Therefore, lots of efforts have been devoted to the development
of ad-hoc methods for EIV systems. Bias eliminating least-squares (BELS) have
been introduced in~\cite{zheng1989unbiased} to correct the bias of standard
least-squares. Another method for the identification of EIV systems has been
obtained by generalizing the so-called Frisch scheme, originally developed for
static EIV models~\cite{frisch1934statistical}, to dynamic models.
Interestingly, the dynamic Frisch scheme provides a unique identified model
under mild conditions~\cite{beghelli1990frisch}. This is in contrast with the
static case, where in general many models are compatible with the observed data.
The accuracy of the Frisch scheme for dynamic system identification has been
extensively studied in the literature~\cite{soederstroem2002perspectives}; the
method has been recently extended to more general noise
setups~\cite{fan2010frisch,ning2015linear,zhang2015errors}. Other EIV
identification methods rely upon maximum-likelihood criteria. Both
time-domain~\cite{soederstroem1981identification,diversi2007maximum} and
frequency-domain~\cite{schoukens1997frequency} approaches have been developed in
the past; for a survey and a comparison of the maximum likelihood methods
see~\cite{soederstroem2010accuracy}.

In this paper, we consider a more general dynamic EIV setup. Specifically, we
assume  that some of the samples may be missing, for instance due to lossy
transmission channels or sensor malfunction. Therefore, our task is to jointly
identify the system and reconstruct the missing input-output values. Some
techniques to deal with this problem have been proposed in the past, both in
time and frequency
domains~\cite{isaksson1993Identification,pintelon2000frequency,wallin2014maximum,markovsky2015identification,zhang2013errors}.
Recently, regularization techniques based on nuclear norm have been proposed for
system identification with missing
data~\cite{liu2013nuclear,markovsky2013structured}.

The method described in this paper to deal with EIV models with missing data
relies on a regularized kernel-based approach. Interpreting kernel-based
regularization as a  Gaussian regression problem~\cite{rasmussen2006gaussian},
we model the unknown impulse response of the system as a zero-mean Gaussian
random vector. The covariance matrix is given by the recently introduced stable
spline kernel~\cite{pillonetto2010new,pillonetto2014kernel}, which penalizes
non-exponentially stable systems. According to the \emph{empirical Bayes} (EB)
paradigm~\cite{maritz1989empirical}, we obtain an impulse response estimator as
a function of the noiseless input and the kernel hyperparameters. These
quantities are estimated maximizing the joint \emph{marginal likelihood}
(marginal likelihood) of the noisy inputs and outputs. We devise an iterative
algorithm to solve the marginal likelihood maximization problem, based on the EM
method. We briefly address the problem of identifiability. We test the proposed
approach with numerical simulations.

In this paper we consider the case of both missing input and missing output
samples, as well as noisy data; as compared to~\cite{pillonetto2009bayesian}
where a kernel-based approach is adopted for the case of noiseless input and
missing output samples.

The paper is organized as follows. In Section~\ref{sec:formulation} we
formulate the problem the identification of EIV models with missing data. In
Section~\ref{sec:sysid} we show how to estimate the system impulse response. In
Section~\ref{sec:ml} we solve the marginal likelihood problem that yields the
missing samples and the kernel hyperparameters. In
Section~\ref{sec:identifiability} we discuss some pitfalls in the model. In
Section~\ref{sec:simulations} we validate our method on a benchmark dataset. In
Section~\ref{sec:conclusions} we discuss our results and conclude the paper.

\subsection{Notation}
We denote by ``$\{a_k\}$'' a sequence of scalars $a_k$ indexed by $k$;
``${\{a_k\}}_{k=a}^{b}$''
is the set of $a_k$ with $k$ ranging from $a$ to $b$. Given
${\{a_k\}}_{k=a}^{b}$,
``$a$'' indicates the column vector of the stacked scalars and
``$a_i$'' indicates the $i$th element of said vector. If $a\in R^{N}$ is a column
vector, ``$\T_{m\times n}(\cdot)$'' indicates the Toeplitz operator that associates
to the vector $a$ the $m\times n$ matrix $A$ such that
\begin{equation}
  {\big[A\big]}_{i,j} =
    \begin{cases}
    a_{i-j+1}, & i\geq j,\, i-j+1\leq N,\\
    0, & \text{otherwise}.
\end{cases}
\end{equation}
The symbol ``$\delta_{i,j}$'' denotes the Kronecker delta and ``$\otimes$'' is the
standard Kronecker product between matrices. The symbol ``$\cong$'' indicates
equality up to an additive constant.

\section{Problem formulation}\label{sec:formulation}
We consider the problem of identifying a dynamic system
from noisy samples of input and output. Figure~\ref{fig:block_scheme} shows a
schematic representation of the setup under study.
\begin{figure}[htb]
  \centering
  \includegraphics[width=0.5\columnwidth]{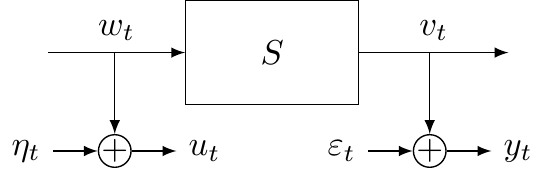}
  \caption{Block scheme of the system setup.}\label{fig:block_scheme}
\end{figure}
The system is strictly causal, asymptotically stable, and linear
time-invariant. The input-output relation can be represented as
\begin{equation}\label{eq:convolution}
  v_t = \sum_{k =1}^\infty g_k  w_{t-k},
\end{equation}
where the $g_k$ is the (unknown) impulse response
of the system. The objective is to reconstruct
the samples of the impulse response from $N$ samples of
the input $u_t$ and output $y_t$. These samples
are measurements of the true system input $w_t$ and output $v_t$, corrupted by
sensor noises
\begin{equation}\label{eq:measurements}
  \begin{split}
    u_t &= w_t + \eta_t,\\
    y_t &= v_t + \varepsilon_t.
  \end{split}
\end{equation}
The noise sequences $\eta_t$ and $\varepsilon_t$ are assumed mutually
independent, Gaussian and white, with unknown variance $\sigma_u^2$ and
$\sigma_y^2$, respectively. The ratio $\gamma=\sigma^2_y/\sigma^2_u$ is assumed
known, in order to guarantee identifiability (see
e.g.~\cite{fernando1985identification}). We suppose that the system is at rest
prior to the collection of the measurements, that is $w_k=0$, $v_k=0$, for all
$k<1$.

We assume also that some of the samples are not available; see the following
example.

\begin{example}
  We have run a system for $N=9$ time instants collecting the following measurements:
\begin{equation*}
  \renewcommand*{\arraystretch}{1.4}
  \begin{matrix}
  \big[ u_1 & \times & u_3 & u_4 & \times & \times & u_7 & \times & u_9\big],\\
  \big[ y_1 & y_2 & \times & y_4 & y_5 & y_6 & \times & \times & y_9\big].
  \end{matrix}
\end{equation*}
\end{example}

Define the set of natural numbers ${\{t^u_k\}}_{k=1}^{N_u}$ such that
$1\leq t^{\,u}_i < t^{\,u}_{i+1}\leq N$, and $u_{t^{u}_i}$ is an
available measurement. In a similar fashion define ${\{t^y_k\}}_{k=1}^{N_y}$.
These sets indicate the $N_u$ and $N_y$ time instants at
which we have available sensor measurements of the input and output
respectively. We
define the available measurement vectors $\tilde u\in \R^{N_u}$ and
$\tilde y\in \R^{N_y}$ such that
\begin{equation}
  \begin{split}
    \tilde u_k &= u_{t^u_k}, \qquad k=1,\dots, N_u,\\
    \tilde y_k &= y_{t^y_k}, \qquad k=1,\dots, N_y.
  \end{split}
\end{equation}
Furthermore, define the operators $\P_u$ and $\P_y$ as the $N_u
\times N$ respectively $N_y\times N$ matrices defined by
\begin{equation}
  {\big[\P_u\big]}_{i,j} = \delta_{i,t^u_j}\,,\quad {\big[\P_y\big]}_{i,j} = \delta_{i,t^{y}_j}\,.
\end{equation}
By construction, these matrices are \emph{right semi-orthogonal}:
\begin{equation}\label{eq:semi-orthogonal}
  \P_u\P_u^T = I_{N_u},\quad
  \P_y\P_y^T = I_{N_y};
\end{equation}
they have full row rank and they represent the mappings betweev the complete
data and the available data:
\begin{equation}
  \tilde u = \P_u u\,, \qquad \tilde y = \P_y y,
\end{equation}
where $u$ and $y$ are vectors of all the stacked values of all (available and
  not) measurements of input and output.
\setcounter{example}{0}
\begin{example}[continued]
  The times of available input measurements are:
  \begin{equation}
    \{t^{u}_k\} = \{1,\,3,\,4,\,7,\,9\}, \quad N_u = 5,
  \end{equation}
  and the $5\times9$ matrix $\P_u$ is
  \begin{equation}
    \P_u =
    \rmatrix{
      1 & 0 & 0 & 0 & 0 & 0 & 0 & 0 & 0\\
      0 & 0 & 1 & 0 & 0 & 0 & 0 & 0 & 0\\
      0 & 0 & 0 & 1 & 0 & 0 & 0 & 0 & 0\\
      0 & 0 & 0 & 0 & 0 & 0 & 1 & 0 & 0\\
      0 & 0 & 0 & 0 & 0 & 0 & 0 & 0 & 1
    }
  \end{equation}
\end{example}
\vspace{0.5em}

We can now formally define the problems of interest in this paper.
\begin{problem}\label{problem:1}[\textbf{System Identification}]
Given $N_u$ ordered samples of the input process $u_t$, collected at times
$t^u_k$, and $N_y$ ordered samples of the output process $y_t$, collected at
times $t^y_k$, estimate the first $n$ samples of the system impulse response
${\{g_k\}}_{k=1}^{n}$.
\end{problem}
We also consider the problem of reconstructing the missing samples:
\begin{problem}\label{problem:2}[\textbf{Input smoothing}]
Given $N_u$ ordered samples of the input process $u_t$, collected at times
$t^u_k$, and $N_y$ ordered samples of the output process $y_t$, collected at
times $t^y_k$, estimate the sample $w_k$, $k\in \N,\, k\leq
N$.
\end{problem}
\begin{problem}\label{problem:3}[\textbf{Output smoothing}]
Given $N_u$ ordered samples of the input process $u_t$, collected at times
$t^u_k$, and $N_y$ ordered samples of the output process $y_t$, collected at
times $t^y_k$, estimate the sample $v_k$, $k\in \N,\, k\leq N$.
\end{problem}

By adopting a kernel-based approach, we introduce a nonparametric model for
the impulse response that allows us to solve the three proposed problems
with a single algorithm based on a marginal likelihood approach. We will first
see how to solve Problem~\ref{problem:1}, using an EB approach.

\begin{remark}
  If $t^u_k=t^y_k=k$ for $k = 1,\dots,N$, then $\tilde u = u$ and $\tilde y =
  y$, and  Problem~\ref{problem:1} corresponds to the standard dynamic
  EIV setup (see~\cite{soederstroem2007errors} for a survey).
\end{remark}

\section{Kernel-based linear system identification}\label{sec:sysid}
We first focus on the problem of identifying $g$. For a given $w_t$,
Problem~\ref{problem:1} becomes a linear regression problem: collecting
$w_t$ into the column vector $w$, we can construct the $N\times n$
Toeplitz matrix $W = \T_{N\times n}(w)$; using this matrix we can write the
convolution~\eqref{eq:convolution} as the matrix product
\begin{equation}
  v = Wg,
\end{equation}
and we can formulate the regression problem in the available output data:
\begin{equation}\label{eq:regression}
  \tilde y = \P_y Wg + \tilde \varepsilon.
\end{equation}
In this equation $\tilde \varepsilon = \P_y\varepsilon$ are the samples of the
noise that correspond to the available samples of the output. From the
semi-orthogonality of $\P_y$, we have that
\begin{equation}
  \E \tilde \varepsilon =  0, \quad \E{\tilde \varepsilon \tilde \varepsilon^T}
  =
  \sigma_y^2 I_{N_y} \,.
\end{equation}
Adopting a kernel-based approach~\cite{pillonetto2014kernel}, we model $g$ as a
Gaussian random vector, with covariance matrix given by a kernel function
suitable for linear system identification. In particular, we use the
\emph{first-order stable-spline kernel}~\cite{pillonetto2010new}, so that
\begin{equation}\label{eq:prior}
  g \sim \mathcal{N}(0,\lambda K_\beta), \quad {\big[K_\beta\big]}_{i,j} := \beta^{ \max(i,j)} \,.
\end{equation}
The quantity $\lambda>0$ is a scaling factor, while $\beta\in(0,\,1)$ is a
shaping parameter that regulates the exponential decay of the realizations from~\eqref{eq:prior}. These two parameters are usually referred to as
\emph{hyperparameters}.

By postulating~\eqref{eq:prior}  a Gaussian prior for the impulse response, we can
derive the joint distribution of the measurements $\tilde y$ and $g$ as
\begin{equation}\label{eq:joint}
  \rmatrix{\tilde y\\g}\sim \mathcal N \left( \rmatrix{0\\0},\,
  \rmatrix {\Sigma_{\tilde y} & \Sigma_{g\tilde y}^T \\ \Sigma_{g\tilde y} & \lambda K_\beta} \right)
\end{equation}
where
\begin{equation}
  \begin{split}
    \Sigma_{\tilde y} &= \lambda \P_y W K_\beta W^T \P_y^T + \sigma_y^2 I_{\tilde N_y}, \\
    \Sigma_{g\tilde y} & = \lambda K_\beta \P_y^T.
  \end{split}
\end{equation}
From~\eqref{eq:joint} we can calculate the posterior distribution of the unknown
impulse response parameters given the available data $\tilde y$ (See, e.g.,~\cite[App. B.7]{soederstroem1981identification}):
\begin{equation}\label{eq:posterior}
g|\tilde y\sim \mathcal N (m, P)
\end{equation}
where
\begin{equation}
  \begin{split}\label{eq:posterior_pars}
    P &= {\left( \frac{1}{\sigma_y^2}W^T\P_y^T\P_y W + {(\lambda
K_\beta)}^{-1}\right)}^{-1}\\
    m &= \frac{1}{\sigma_y^2}P W^T\P_y^T \tilde y.
  \end{split}
\end{equation}
With the posterior distribution, we find the (Bayesian) minimum variance
estimate of $g$ as the posterior mean $m$.
From~\eqref{eq:posterior_pars}, we see that the posterior mean depends on the quantities
$\lambda$, $\beta$ and $w$, as well as on the output noise covariance
$\sigma_y^2$. All these parameters are unknown and need to be estimated from the data. Using an
EB approach, we estimate the parameters by replacing them with their maximum
marginal likelihood
estimates $\hat\lambda$, $\hat\beta$ and $\hat w$ (and $\hat \sigma_y^2$, but
this needs a special treatment: see Section~\ref{sec:identifiability}). In the
next section we focus on the problem of finding the maximizers
of the marginal likelihood. Solving the marginal likelihood problem is also the key to solve
Problem~\ref{problem:2} and Problem~\ref{problem:3}.

\section{Kernel-based input and output smoothing}\label{sec:ml}
\subsection*{Input smoothing and hyperparameter selection}
Consider the measurement model~\eqref{eq:measurements}. We can write it as a
regression in the smoothed input $w_t$ by observing that
\begin{equation}
  v = Wg = Gw
\end{equation}
where $G = \T_{N\times N}(g)$. Considering also the unavailability of some
data, we obtain the linear regression model in the available measurements
\begin{equation}
  \rmatrix{\tilde y \\ \tilde u }= \rmatrix{\P_y G\\ \P_u }  w +
  \rmatrix{\tilde \varepsilon\\ \tilde \eta}
\end{equation}
where $\tilde \eta = \P_u\eta$ is white noise of variance $\sigma_u^2$.
Under the Bayesian prior assumption~\eqref{eq:prior}, the available observation
vector and the impulse response parameter are jointly Gaussian, with a
log-likelihood given by
\begin{equation}\label{eq:complete_likelihood}
  \begin{split}
  L(\tilde y,\tilde u,g;w,\lambda,\beta) &= \\
  \log p(\tilde y|g;w) + \log p&(\tilde u;w) + \log p(g;\lambda,\beta) \,,
  \end{split}
\end{equation}
where
\begin{equation}
  \begin{split}
    & \log p(\tilde y|g;w) \cong -\frac{1}{2\sigma_y^2}\|\tilde y -\P_y Gw\|^2\\
    & \log p(\tilde u|w) \cong -\frac{1}{2\sigma_u^2}\|\tilde u - \P_u w\|^2\\
    & \log p(g;\lambda,\beta) =  -\frac{1}{2}g^T{(\lambda K_\beta)}^{-1}g -
    \frac{1}{2}\log\det\left( \lambda K_\beta \right)
  \end{split}
\end{equation}
Since $g$ is not available, we can interpret it as a \emph{latent variable},
which we estimate using the expectation maximization (EM) method.  The term
``EM method'' refers to a class of algorithms used to solve maximum likelihood
problems with latent variables.  In these methods, an iterative algorithm is
built by alternating between estimating the likelihood, and
updating the likelihood parameters using the estimated likelihood.

The estimated likelihood is created in the \emph{expectation step} by taking the conditional
expectation of the log-likelihood with respect to the posterior distribution of
the latent variables given the available data, for some estimate of the parameters:
\begin{equation}\label{eq:Q}
  Q\big(\theta;\theta^{(k)}\big) = \E\big[ L(\tilde y,\tilde u,g;\theta) \big]
\end{equation}
where $\theta = [w^T,\,\lambda,\,\beta]$ are the parameters to be estimated, and
the expectation is taken with respect of the distribution~\eqref{eq:posterior},
with the vector $\theta$ replaced with an estimate $\theta^{(k)}$. By
construction this function
is such that
\begin{equation}
  L\big(\tilde y,\tilde u;\theta\big) -
  L\big(\tilde y,\tilde u;\theta^{(k)}\big) \geq
  Q\big(\theta;\theta^{(k)}) - Q\big(\theta;\theta^{(k)}\big).
\end{equation}
where $L\big(\tilde y,\tilde u;\theta\big)$ is the marginal likelihood of the available data.
In the subsequent \emph{maximization step}, the parameter update is chosen as
the maximum of $Q(\theta;\theta^{(k)})$, so that
\begin{equation}
  Q\big(\theta^{(k+1)};\theta^{(k)}\big) - Q\big(\theta;\theta^{(k)}\big)> 0
\end{equation}
and consequently the marginal likelihood, in the updated parameters, is increased as well.
By iterating the expectation and maximization steps, from any initialization of
the parameters, we obtain a sequence of estimates of the parameters that
converge to a local maximizer of the marginal likelihood of the available data (for a
complete look on the EM algorithm~\cite{mclachlan2007algorithm}).

In the case at hand, we can compute the expectation of the joint log-likelihood~\eqref{eq:complete_likelihood} in closed form:
\begin{equation*}
  \begin{split}
    & \E\log p(\tilde y|g;w) \cong \frac{1}{\sigma_y^2}\left(\tilde y^T \P_y
    M^{(k)}w - \frac{1}{2}w^T A^{(k)} w\right),\\ & \E\log p(\tilde u|w) \cong
    \frac{1}{\sigma_u^2}\left(\tilde u^T\P_u w  - \frac{1}{2}w^T\P_u^T\P_u
    w\right),\\ & \E\log p(g;\lambda,\beta) =
      -\frac{1}{2}\trace\Big\{{\big(\lambda
    K_\beta\big)}^{-1}(P^{(k)}+m^{(k)}m^{(k)\,T})\Big\} \\ &\hspace{8em}-
\frac{1}{2}\log\det\left( \lambda K_\beta \right).
\end{split}
\end{equation*}
All the expectations are taken with respect to the posterior
distribution~\eqref{eq:posterior}, with $\lambda$, $\beta$, and $w$ replaced by
their estimates $\lambda^{(k)}$, $\beta^{(k)}$ and $w^{(k)}$. In these
expressions $m^{(k)}$ and  $P^{(k)}$ are the posterior mean and covariance as
expressed in~\eqref{eq:posterior_pars}, with the unknown parameters replaced by
their estimates; $M^{(k)}=T_{N\times N}(m^{(k)})$ is the Toeplitz matrix of the
posterior mean and $A^{(k)}$ is the posterior second moment of the matrix
$G^T\P_y^T=T_{N\times N}{(g)}^T \,\P_y^T$, that is
\begin{equation}
  \begin{split}
    A^{(k)} &= \E\big\{G^T\P_y^T\P_y G \big\}\\ &=
  \devec^T\Big[(P^{(k)} + m^{(k)}m^{(k)\,T}\big) \otimes \P_y^T\P_y\Big] \devec
\,.  \end{split}
\end{equation}
The $Nn\times N$ matrix $\devec$ is defined as
\begin{equation}
  \devec^T = \rmatrix{I_N & S & S^2 & \cdots & S^n}
\end{equation}
where $S$ is the $N\times N$ upward shift operator
\begin{equation}
  \label{eq:shift_op} {\big[S\big]}_{i,j} = \delta_{i,j-1}.
\end{equation}

To find the updated parameter values, we maximize the conditional expectations
with respect to the parameters:
\begin{align}
    & w^{(k+1)} = \arg\max_w\, \E\log p(\tilde y|g;w)+ \E\log p(\tilde u
|w)\,,\label{eq:update_w} \\
    &\lambda^{(k+1)}, \beta^{(k+1)} = \arg\max_{\lambda,\beta}\,\E\log
    p(g;\lambda,\beta)\,. \label{eq:update_hypers}
\end{align}

The cost function in~\eqref{eq:update_w} is quadratic in the decision variable,
so the maximum is available in closed form as
\begin{equation*}
  w^{(k+1)} ={\left(A^{(k)}+\gamma\,
    \P_u^T\P_u  \right)}^{-1}  \left(
    M^{(k)\,T}\P_y^T\tilde y +
  \gamma\, \P_u^T\tilde u \right),
\end{equation*}
where $\gamma = \sigma_y^2/\sigma_u^2$.
The optimization~\eqref{eq:update_hypers} can be solved in closed form with
respect to $\lambda$:
\begin{equation}
  \lambda^*(\beta) = \frac{1}{n}\trace\Big\{K_\beta^{-1}
  \big(P^{(k)}+m^{(k)}m^{(k)\,T}\big)\Big\}.
\end{equation}
With this, the update of $\beta$ is given by
\begin{equation}
  \beta^{(k+1)} =\arg\max_{\beta\in(0,\,1)}\, n\log\big( \lambda^*(\beta) \big) +
  \log\det K_\beta\, ,
\end{equation}
which can be solved with scalar optimization methods, or grid search. Once
we have $\beta^{(k+1)}$, we also have the update for~$\lambda$:
\begin{equation}
  \lambda^{(k+1)} = \lambda^*\big(\beta^{(k+1)}\big).
\end{equation}

Appealing to theory of the EM-method, we have the following result:
\begin{theorem}\label{thm:1}
  The sequences $\{w^{(k)}\}$, $\{\lambda^{(k)}\}$, and $\{\beta^{(k)}\}$
  generated by the iterations~\eqref{eq:update_w} and~\eqref{eq:update_hypers}
  are such that:
  \begin{equation}\label{eq:proof1}
    L(\tilde y,\tilde u;w^{(k+1)},\lambda^{(k+1)},\beta^{(k+1)}) >
    L(\tilde y,\tilde u;w^{(k)},\lambda^{(k)},\beta^{(k)}).
  \end{equation}
  where $L(\tilde y,\tilde u;w,\lambda,\beta)$ is the marginal likelihood of
  the data; and
  \begin{equation} \label{eq:proof2}
    L(\tilde y,\tilde u;w^{(k)},\lambda^{(k)},\beta^{(k)}) \to
    L^*
  \end{equation}
  as $k\to \infty$,  where $L^*$ is a local extremum of $L(\tilde y
  ,\tilde u; w,\lambda,\beta)$.
\end{theorem}
\begin{proof}
  By construction, the iterations~\eqref{eq:update_w} and~\eqref{eq:update_hypers} are iterations in an EM algorithm. The E-step function~\eqref{eq:Q}, seen as the function of two variables $Q(x;y)$ is continuous in
  $x$ and $y$. So the sequence generated satisfy the conditions in Theorem~1 in~\cite{wu1983convergence} and~\eqref{eq:proof1} and~\eqref{eq:proof2} follow.
\end{proof}
Interestingly, except for pathological cases, the EM-method is guaranteed to
converge to a local maximum of the marginal likelihood (see~\cite[Ch. 3]{mclachlan2007algorithm} for details).
\begin{corollary}\label{cor:1}
  If the sequences $\{w^{(k)}\}$, $\{\lambda^{(k)}\}$, and
  $\{\beta^{(k)}\}$
  generated by the iterations~\eqref{eq:update_w} and~\eqref{eq:update_hypers}
  are such that:
  \begin{equation*}
    \|w^{(k+1)} - w^{(k)}\|^2_2+\|\beta^{(k+1)} - \beta^{(k)}\|^2_2 +
    \|\lambda^{(k+1)} - \lambda^{(k)}\|^2_2 \to 0,
  \end{equation*}
  as $k\to \infty$, then they converge to a stationary point of $L(\tilde y,
  \tilde u; w ,\lambda,\beta)$.
\end{corollary}
\begin{proof}
  Follows directly from Theorem~6 in~\cite{wu1983convergence}.
\end{proof}
\begin{remark}
  Theorem~\ref{thm:1} gives a natural stopping criterion for the EM algorithm.
  When the increase in the likelihood between two iterates is below a certain
  threshold, (approximate) convergence to a maximimum is safely guaranteed.
  Corollary~\ref{cor:1} further guarantees the convergence of the parameters to
  a local maximizer when the change between iterations is infinitesimal.
\end{remark}
\subsection{Kernel-based output smoothing}\label{sec:output}
To solve Problem~\ref{problem:3} we first observe that, since the output noise
is white, the output smoothing problem is a simulation problem, and the
smoothed output signal is given by the convolution $Wg$.
After solving Problem~\ref{problem:1} and Problem~\ref{problem:2}, we can find
an estimate of the smoothed output signal $\hat v$ by plugging in the
estimates $\hat w$ and $\hat g$ in the convolution, obtaining
\begin{equation}
  \hat v= \hat W\hat g.
\end{equation}

\section{Some remarks on identifiability}\label{sec:identifiability}
It is well known (see, e.g.~\cite{bottegal2011identifiability},~\cite{soederstroem2002perspectives}) that, in general,
errors-in-variables problems are not identifiable. Different
models may explain the same observed data and therefore it is impossible to
assess the validity of a certain model from the data. In the case of Gaussian
noise, where only
second moments carry information about the distributions, any attempt to
identify the noise variances, the system, and the input samples is bound to
fail.  In our EM framework, this follows from the shape of the likelihood~\eqref{eq:complete_likelihood}: for instance leaving free both $\sigma_u^2$ and
$w$,  we can choose $w=\tilde u$ and $L(\tilde{y},\tilde{u},g;w,\lambda,\beta)$
can be made arbitrarily large by choosing a small enough $\sigma_u^2$. Various
additional assumptions can be posed to circumvent the non-identifiability
issue, see~\cite{soederstroem2007errors}. In our setup, if we know the ratio
$\gamma = \sigma_y^2/\sigma_u^2$ we can estimate the unknown variances by
adding the following equation to the iterations of the EM method: \begin{align}
  \sigma_y^{2\,(k+1)} &= \frac{\tilde y^T\tilde y - 2\tilde y^T
  \P_y M^{(k)}w^{(k+1)} + w^{(k+1)\,T}A^{(k)}w^{(k+1)}}{N_u+N_y}\nonumber\\ &
  + \gamma\frac{\tilde u^T\tilde u - 2\tilde u^T\P_u w^{(k+1)} +
  w^{(k+1)\,T}\P_u^T\P_u w^{(k+1)}}{N_u+N_y},\nonumber\\
  \sigma_u^{2\,(k+1)} &=
\sigma_y^{2\,(k+1)}/\gamma.\label{eq:noisecov} \end{align}

In the case of missing data we have other identifiability problems, in
addition to the ones inherited from errors-in-variables. The possibility of
multiple models explaining the available data is linked to aliasing, as the
missing data can be seen as data decimation~\cite{wallin2002multiple}. In order
to have a unique solution to the likelihood problem
\begin{equation}\label{eq:minimize_w}
  \underset{w,\,\lambda,\,\beta}{\text{maximize}}\; \log p(\tilde y|g;w) + \log
  p(\tilde u;w) + \log p(g;\lambda,\beta), \end{equation} where $g$ is the true
  impulse response, we need that the symmetric matrix
  \begin{equation}\label{eq:A}
  \frac{1}{\sigma^2_y}G^T\P_y^T\P_y G\!+\!\frac{1}{\sigma_u^2} \P_u^T\P_u
\end{equation} is invertible. This is the case as long as the effect of every
missing input sample is visible at least once in the output:
\begin{proposition}\label{prop:identifiability}
  There is a unique solution to~\eqref{eq:minimize_w} if and
  only if for every missing input sample time $\tau_u \notin \{t^u_i\}$, there
  is a $k \in \{0\dots n\}$ and a $\tau_y \in \{t^y_i\}$ such that $g_k\neq 0$
  and $k+\tau^u = \tau_y$.  \end{proposition} \begin{proof} Matrix~\eqref{eq:A} is
  invertible iff there is no $\alpha$ such that $\P_y G \alpha=0$ and
  $\P_u\alpha=0$. The condition $\P\alpha =0$ means that $\alpha$ can be
  written as $\sum_{i\notin \{t^u_k\}} \alpha_i e_i$ where $e_i$ are vectors in
  the canonical basis of $\R^N$ and $a_i$ are scalars. The condition
  $\P_y G\alpha\neq 0$ translates into $\sum_{i\notin \{t^u_k\}} a_i \P_y S^i g
  \neq 0$, where $S$ is defined in~\eqref{eq:shift_op}. This concludes the
  proof.  \end{proof}
g
\section{Simulations}\label{sec:simulations}
To evaluate the performance of the proposed method, we perform a set of Monte
Carlo (MC) simulations. In the MC simulations, we identify the impulse
responses of 500 systems from the dataset D1 described in~\cite{chen2012sparse}.
For each system in the dataset, we generate $N  = 210$ input and output
samples; the input is Gaussian white noise with variance equal to 1. The output
measurements are affected by Gaussian white noise of variance equal 0.1, namely
10\% of the noiseless output variance. The variance of the noise affecting the
input varies with the experiment.

We use the iterative method presented in Section~\ref{sec:ml} to estimate the
first $n=100$ samples of the impulse response. The noise variance is updated
iteratively with~\eqref{eq:noisecov}. The iterations are initialized at
$w^{(0)}=u$, $\lambda^{(0)}=10$, $\beta^{(0)}=0.6$. The noise variances
$\sigma_y^2$ and $\sigma_u^2$ are initialized, respectively, at the sample
variance
of the least squares residuals and at $\sigma_y^{2\,(0)}/\gamma$. The iterations
are stopped when the relative change of the parameter updates is below 1\%.

We evaluate the goodness of fit using the standard score
\begin{equation}
  \mathrm{fit}(a, a_\mathrm{ref}) = 1- \frac{\| a -
  a_\mathrm{ref}\|_2}{\|a_\mathrm{ref} - \mathrm{mean}(a_\mathrm{ref})\|_2},
\end{equation}
where $a_\mathrm{ref}$ is a true value and $a$ its estimate. We calculate
the median fit of the estimated impulse responses, inputs, and outputs over the
dataset.

We consider two different scenarios. In the first scenario, we corrupt the
dataset with increasing fractions of missing samples. In the second scenario, we
corrupt the dataset with input noises of increasing variance.
Table~\ref{tab:summary} gives a summary of the experimental conditions.

\begin{table}[b]
  \centering
  \begin{tabular}{ccccc}
    \toprule
    & $\sigma_y^2$ & $\sigma_u^2$ & missing input  & missing output\\
    \midrule
    A (Exp. 1) & 10\% & 10\% &  $0\%\div50\%$ & $0\%$\\
    A (Exp. 2) & 10\% & 10\% & $0\%$ & $0\%\div50\%$ \\
    A (Exp. 3) & 10\% & 10\% & $0\%\div25\%$ & $0\%\div25\%$\\
    B & 10\% & $0\%\div100\%$ & $0\%$ & $0\%$\\
    \bottomrule
  \end{tabular}
  \caption{Experimental conditions in the simulation scenarios}\label{tab:summary}
\end{table}

\subsection*{Scenario A\@: Missing data}
The input noise is Gaussian white noise with variance 0.1 (10\% of the input
signal variance). Before performing
the identification, we randomly select and remove a fraction of the available
data: in Exp. 1, we remove from 0\% to 50\% of the input samples, in 10\%
increments; in Exp. 2, we remove from 0\% to 50\% of the output samples,
in 10\% increments; in Exp. 3, we remove equal fractions of input and output,
between 0\% and 25\%, in 5\% increments.  The results are plotted in
Figure~\ref{fig:scenario_a}. Interestingly, a large fraction of missing input
samples has severe effect on the performance, whereas a large fraction of
missing output samples has a milder effect on the identification performance. In
Exp. 1 and Exp. 2, the model has always resulted identifiable, whereas in Exp. 3
a number of systems were non-identifiable. The results are collected in
Table~\ref{tab:unsolve}.
\begin{figure}[H]
  \centering
  \includegraphics[width=\columnwidth]{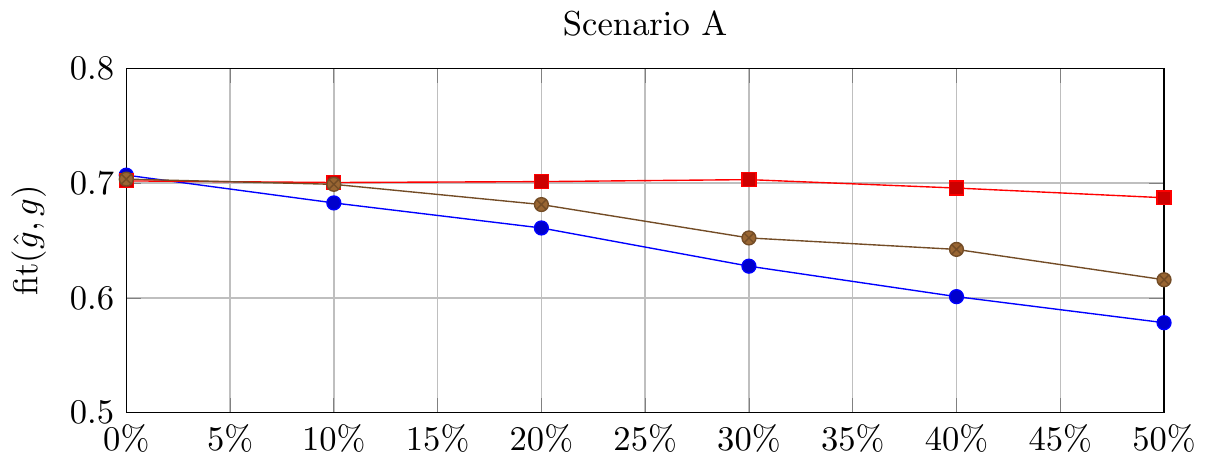}
  \includegraphics[width=\columnwidth]{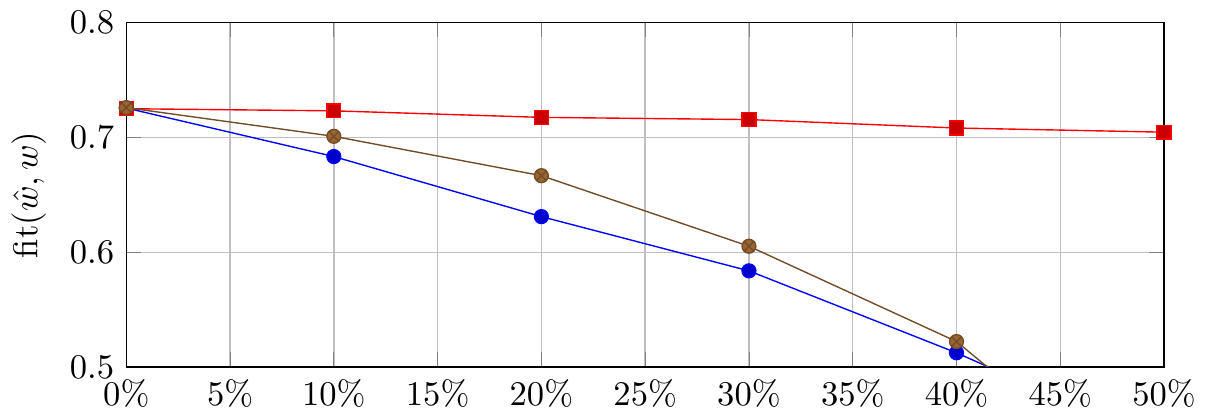}
  \includegraphics[width=\columnwidth]{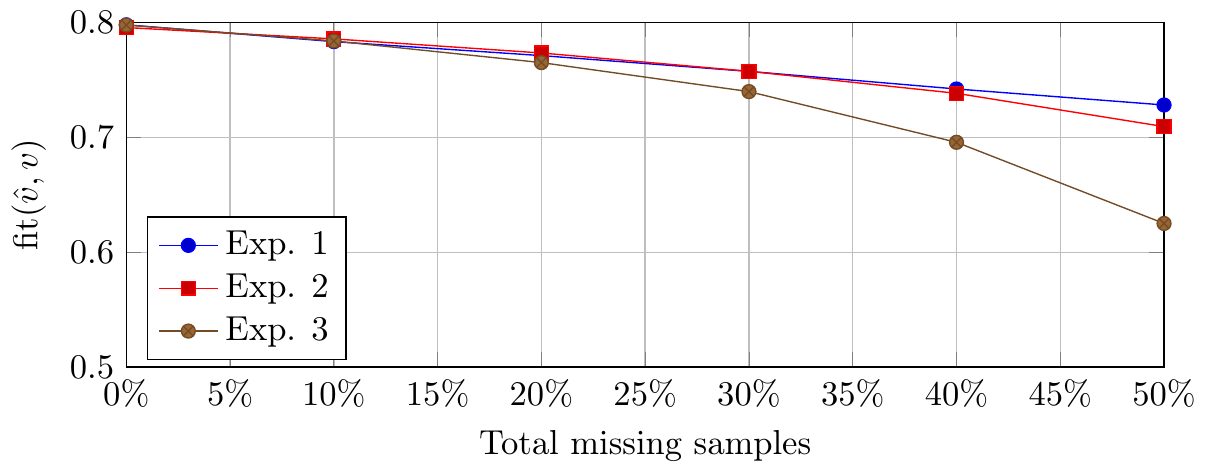}
  \caption{Plot of the median fit of the impulse response (top), the smoothed
    input (middle) and the smoothed output (bottom) over 500 MC runs, for increasing
    fractions of missing samples; In Exp.1  we remove input samples, in
    Exp.2  we remove output samples, in Exp.3  we remove input and
    output samples.
}\label{fig:scenario_a}
\end{figure}

\begin{table}[b]
  \centering
  \begin{tabular}{ccccccc}
    \toprule
    Unsolvable problems&     0  &   3  &  5 &  16 &  27 &  40\\
    \midrule
    Total missing samples & 0\% & 10\% & 20\% & 30\% & 40\% & 50\%\\
    \bottomrule
  \end{tabular}
  \caption{Number of non-identifiable systems in Exp. 3 (out of 500)}\label{tab:unsolve}
\end{table}

\subsection*{Scenario B\@: Input noise}
The input noise is Gaussian white noise. We consider values of the input noise
variance between 0 (no noise) and 1 (same variance as the input), in increments
of 0.2. The results are plotted in Figure~\ref{fig:scenario_b}. In this
scenario, we compare the performance of the proposed method with a kernel-based
identification method that does not account for input noise. We estimate the
impulse response using the posterior mean $m$ from~\eqref{eq:posterior_pars},
with $\lambda$, $\beta$ and $\sigma_y^2$ estimated trough marginal likelihood,
with all instances of $w$ replaced by the noisy measurements $u$.
\begin{figure}[htb]
  \centering
  \includegraphics[width=\columnwidth]{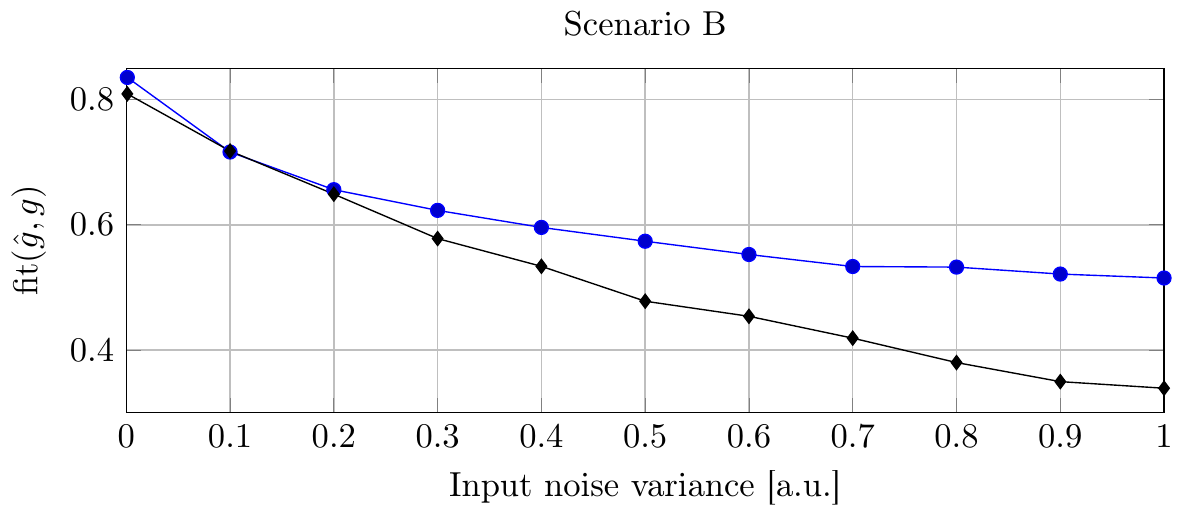}
  \caption{Plot of the median fit of the impulse response estimate over 500 MC
    runs, for increasing values of the input noise variance. We compare the
    proposed estimator (blue) with performance of an estimator that does not
  account for input noise (black).}\label{fig:scenario_b}
\end{figure}

\section{Conclusions}\label{sec:conclusions}
In this paper we have presented a nonparametric kernel-based method for the
identification of
the impulse response of a linear systems in the presence of noisy and missing input-output
data. The method relies upon a Gaussian regression framework,
where the impulse response of the system is modeled as a Gaussian process with
a suitable covariance matrix. Using an EB approach, we find the minimum
mean-squared estimate of the impulse response. This estimate depends on the
unknown noiseless input, as well as on the kernel hyperparameters and the noise
variances. These quantities are estimated from the marginal likelihood of the data, obtained
integrating out the impulse response. We have devised an iterative scheme that
solves the marginal likelihood maximization in simple updates, and we have discussed
the convergence properties of the algorithm. We have tested the method on a data
bank of linear systems, where we have analyzed the degradation in performance
for increasing amounts of missing data, and increasing noise variance on the input
measurements. We have briefly addressed the question of identifiability; and
simulations seem to validate our theoretical results, however, this aspect
needs further study.

\bibliographystyle{ieeetr}
\bibliography{main.bbl}

\end{document}